\documentclass[sigconf]{acmart}

\usepackage{bbm}
\usepackage{graphicx}
\usepackage{subfigure}
\usepackage{algorithm, algorithmicx, algpseudocode}
\usepackage{url}
\usepackage{xcolor}
\usepackage{enumerate}
\usepackage{multirow}
\usepackage{tablefootnote}
\usepackage[hang,flushmargin]{footmisc}
\usepackage{tikz}
\usepackage{balance}

\algnewcommand\algorithmicinput{\textbf{Input:}}
\algnewcommand\Input{\item[\algorithmicinput]}
\algnewcommand\algorithmicoutput{\textbf{Output:}}
\algnewcommand\Output{\item[\algorithmicoutput]}

\newtheorem{remark}{Remark}
\newcommand{\eat}[1]{}

\newcommand{\R}{\mathbb R}
\newcommand{\A}{\mathbb A}

\renewcommand{\O}{{\mathbb O}}
\newcommand{\Z}{{\mathbb Z}}
\newcommand{\M}{{\mathbb M}}

\newcommand{\I}{{\mathbb I}}
\newcommand{\D}{{\mathcal D}}
\newcommand{\PP}{{\mathcal P}}
\newcommand{\LL}{{\mathcal L}}

\DeclareMathOperator*{\E}{\mathbb{E}}
\DeclareMathOperator*{\argmin}{\arg\min}

\DeclareMathOperator*{\supp}{supp}
\newcommand{\cC}{\mathcal{C}}
\newcommand{\Csol}{\mathcal{C}_{\mathsf{sol}}}
\newcommand{\Copt}{\mathcal{C}_{\mathsf{opt}}}
\newcommand{\opt}{\mathsf{opt}}
\newcommand{\sol}{\mathsf{sol}}

\newcommand{\abs}[1]{\left| #1 \right|}

\newcommand{\eps}{\epsilon}

\newcommand{\kt}{(k,t)}

\allowdisplaybreaks

\begin{document}
\title{Distributed Partial Clustering}
\thanks{Sudipto Guha was supported in part by NSF award 1546151. Qin Zhang was supported in part by NSF CCF-1525024 and
IIS-1633215.}



\author{Sudipto Guha}
\affiliation{
	\institution{University of Pennsylvania}
	\city{Philadelphia}
	\state{PA}
	\postcode{19104}
	\country{United States}}
\email{sudipto@cis.upenn.edu}

\author{Yi Li}
\affiliation{
	\institution{Nanyang Technological University}
	\country{Singapore}}
\email{yili@ntu.edu.sg}

\author{Qin Zhang}
\affiliation{
	\institution{Indiana University Bloomington}
	\city{Bloomington}
	\state{IN}
	\postcode{47401}
	\country{United States}}
\email{qzhangcs@indiana.edu}

\copyrightyear{2017} 
\acmYear{2017} 
\setcopyright{acmlicensed}
\acmConference{SPAA '17}{July 24-26, 2017}{Washington DC,
USA}\acmPrice{15.00}\acmDOI{10.1145/3087556.3087568}

\begin{abstract}
Recent years have witnessed an increasing popularity of algorithm design for distributed data, largely due to the fact that massive datasets are often collected and stored in different locations. In the distributed setting communication typically dominates the query processing time. Thus it becomes crucial to design communication efficient algorithms for queries on distributed data.
Simultaneously, it has been widely recognized that partial optimizations, where we are allowed to disregard a small part of the data, provide us significantly better solutions. The motivation for disregarded points often arise from noise and other phenomena that are pervasive in large data scenarios.

In this paper we focus on partial clustering problems, $k$-center, $k$-median and $k$-means, in the distributed model, and provide algorithms with  communication sublinear of the input size. As a consequence we develop the first algorithms for the partial $k$-median and means objectives that run in subquadratic running time. 
We also initiate the study of distributed algorithms for clustering uncertain data, where each data point can possibly fall into multiple locations under certain probability distribution.  
\end{abstract}

\settopmatter{printacmref=false, printfolios=false}
\fancyhead{}

\maketitle

\section{Introduction}
\label{sec:intro}

The challenge of optimization over large quantities of data has brought communication efficient {\em distributed} algorithms to the fore.  From the perspective of optimization, it has also become clear that {\em partial optimizations}, where we are allowed to disregard a small part of the input, enable us to provide significantly better optimization solutions compared with those which are forced to account for the whole input. While several algorithms for distributed clustering have been proposed, partial optimizations for clustering problems, introduced by Charikar et al.~\cite{CKMN01}, have not received as much attention. While the results of Chen~\cite{Chen08} improve the approximation ratios, the running time of the $k$-median and $k$-means versions have not been improved and the (at least) quadratic running times have remained as a barrier.

In this paper we study partial clustering under the standard $(k,t)$-median/means/center objective functions, where $k$ is the number of centers we can use and $t$ is the maximum number of points we can ignore. In the distributed setting, let $s$ denote the number of sites. 
The $(k, t)$-center problem has recently been studied by Malkomes et al.~\cite{Malkomes}, who gave a $2$-round $O(1)$-approximation algorithm with $\tilde{O}(sk + st)$ bits of communication\footnote{We hide $\text{poly}\log n$ factors in the $\tilde{O}$ notation.}, assuming that each point can be encoded in $O(1)$ bits.   In fact, we observe that results from streaming algorithms \cite{GMMMO03} can in fact provide us $1$-round $O(1)$-approximation algorithms with $\tilde{O}(sk+st)$ bits of communication for $(k,t)$-center, $(k,t)$-median, and $(k,t)$-means.  However, in many scenarios of interest, we have $n > t \gg k$ and $t \gg s$.  Thus the $st$ term generates a significant communication burden. In this paper we reduce $\tilde{O}(st)$ to $\tilde{O}(t)$ for the $(k,t)$-center problem, as well as for $(k,t)$-median and $(k,t)$-means problems and unify their treatment. We also provide the first subquadratic algorithms for median and means version of this problem.

Large data sets often have erroneous values. Stochastic optimization has recently attracted a lot of attention in the field of databases, and has substantiated as a subfield called `uncertain/probabilistic databases' (see, e.g.,~\cite{suciu}). For the clustering problem, a method of choice is to first model the underlying uncertainty and then cluster the uncertain data.  Clustering under uncertainty has been studied in centralized models \cite{CM08,GM09}, but the algorithms proposed therein do not consider communication costs. Note that it typically requires significantly more communication to communicate a distribution (for an uncertain point) than a deterministic point, and thus black box adaptations of centralized algorithms do not work well in the distributed setting.  In this paper we propose communication-efficient distributed algorithms for handling {\em both} data uncertainty and partial clustering.  To the best of our knowledge neither distributed clustering of uncertain data nor partial clustering of uncertain data has been studied. We note that both problems are fairly natural, and likely to be increasingly useful as distributed cloud computing becomes commonplace.

\medskip
\noindent{\bf Models and Problems.} 
We study the clustering problems in the {\em coordinator} model, in which there are $s$ sites and one central coordinator, who are connected by a star communication network with the coordinator at the center.  However, direct communication between sites can be simulated by routing via the coordinator, which at most doubles the communication. The computation is in terms of rounds. At each round, the coordinator sends a message (could be an empty message) to each site and every site sends a message (could be an empty message) back to the coordinator.  The coordinator outputs the answer at the end. The input $\A$ is partitioned into $(\A_1, \ldots, \A_s)$ among the $s$ sites. Let $n_i = \abs{\A_i}$, and $n= \abs{\A} = \sum_{i \in [s]} n_i$ be the total input size.

We will consider clustering over a graph with $n$ nodes and an oracle distance function $d(\cdot,\cdot)$. An easy example of such is points in Euclidean space. More complicated examples correspond to documents and images represented in a feature space and the distance function is computed via a kernel.  We now give the definitions of $(k, t)$-center/median/means.

\begin{definition}[$(k,t)$-center,median,means]
\label{def:clustering}

Let $\A$ be a set of $n$ points and $k$, $t$ are integer parameters ($1\leq k\leq n$, $0\leq t\leq n$).  In the $(k,t)$-median problem we want to compute 
$$\min_{K, \O \subseteq \A} 
\sum_{p\in \A\setminus \O} d(p, K) \quad \text{subject to} \quad \abs{K} \leq k \ \text{ and } \abs{\O} \leq t,$$ where $d(p, K) = \min_{x\in K} d(p, x)$.  We typically call $K$ the {\em centers} and $\O$ the {\em outliers}.  In the $(k,t)$-means and the $(k,t)$-center problem we replace the objective function $\sum_{p\in \A\setminus \O} d(p, K)$ with $\sum_{p\in \A\setminus \O} d^2(p, K)$
and $\max_{p\in \A\setminus \O} d(p, K)$ respectively.
\end{definition}

In the definition above, we assume that centers are chosen from the input points. In the Eucldiean space, such restriction will only affect the approximation by a factor of $2$.

For the uncertain data, we follow the assigned clustering introduced in \cite{CM08}. Let $\PP$ be a finite set of points in a metric space. There are $n$ input nodes $\A$, where node $j$ follows an independent distribution $\D_j$ over $\PP$.  Each site $i$ knows the distributions $\D_j$ associated with the nodes $j\in \A_i$. 

\begin{definition} [Clustering Uncertain Data]
\label{def:uncertain-clustering}  
In clustering with uncertainty, the output is a subset $K\subseteq \PP$ of size $k$ (centers), a subset $\O\subseteq \PP$ of size at most $t$ (ignored points), as well as a mapping $\pi:\A\to K$. In every realization $\sigma:\A\to \PP$ of the values of the input nodes, node $j\in \A$ (now realized as $\sigma(j)\in \PP$) is assigned to the same center $\pi(j)\in K$.  In uncertain $(k,t)$-median, the goal is to minimize the expected cost
\begin{equation}
\label{aaa02}
\E_{\sigma\sim \prod_{j\in\A} \D_j} \left[\sum_{j\in \A\setminus \O} d(\sigma(j),\pi(j))\right] = \sum_{j\in \A\setminus \O} \E_{\sigma\sim\D_j} \left[d(\sigma(j),\pi(j))\right].
\end{equation}
The definition of uncertain $(k,t)$-means is basically the same as uncertain $(k,t)$-median, except that we replace the objective function (\ref{aaa02}) with $\sum_{j\in \A\setminus \O} \E_{\sigma\sim\D_j} \left[d^2(\sigma(j),\pi(j))\right]$.
For uncertain $(k,t)$-center, we have two objectives:
\begin{eqnarray}
\max_{j\in \A\setminus \O}  \left( \E_{\sigma\sim \D_j} \left[d(\sigma(j),\pi(j)) \right]\right)
\qquad \label{kcensimple}
\\
\E_{\sigma\sim \prod_j \D_j} \left[\max_{j\in \A\setminus \O} d(\sigma(j),\pi(j))\right] \qquad \label{kcenhard}
\end{eqnarray}

Note that these two objectives are {\em not} equivalent, since $\E$ and $\max$ do not commute in Equation~\eqref{kcenhard} and we cannot equate it to \eqref{kcensimple}. Equation~\eqref{kcensimple} is in the same spirit as Equation~\eqref{aaa02}, and corresponds to a {\em per point} measurement. We term this problem as uncertain $(k,t)$-center-pp. Equation~\eqref{kcenhard} corresponds to a more {\em global} measurement and we term this problem as uncertain $(k,t)$-center-g. This version was considered in \cite{CM08,GM09}.
\end{definition}

\begin{table*}[t]
\centering
{\scriptsize
\begin{tabular}{|c|c|c|c|c|c|c|}
\hline
Objective &  Approx. &  Centers & Ignored & Rounds & Total Comm. & Local Time\\
\hline
\rule{0pt}{3ex}
\multirow{2}{*}{median} & $O(1)$ & \multirow{2}{*}{$k$} & $t$ & $ \multirow{2}{*}{2}$ & $\tilde{O}((sk + t)B)$ & 
 $\tilde{O}(n^2_i)$, $\tilde{O}(k^2t^2(sk+t)^3)$ \\
 & $O(1 + 1/\eps)$ & & $(1+\eps)t$ & & $\tilde O((sk + t)B)$  & $\tilde{O}(n^2_i)$, $\tilde{O}((sk+t)^2)$  \\
\hline
\rule{0pt}{3ex}
means & $O(1 + 1/\eps)$ & $k$ & $(1+\eps)t$ & 2 & $\tilde O((sk + t)B)$  & $\tilde{O}(n^2_i)$, $\tilde{O}((sk+t)^2)$  \\
\hline 
\rule{0pt}{3ex}
center & $O(1)$ & $k$ & $t$ & 2 & $\tilde O((sk + t)B)$  & $\tilde{O}((k+t)n_i)$, $\tilde{O}((sk+t)^2)$  \\
\hline 
\begin{tabular}[x]{@{}c@{}}\rule{0pt}{2.5ex} uncertain\\
median/\\means/\\center-pp \end{tabular} & 
\multicolumn{5}{c|}{{as in the regular case above}} & {$+ O(n_i T)$, unchanged}\\
\hline
\rule{0pt}{3ex}
center-g & $O(1 + 1/\eps)$ & $k$ & $(1+\eps)t$ & 2 & $\tilde O(skB + tI + s\log \Delta)$  & $\tilde{O}(n^2_i \log \Delta)$, $\tilde{O}((sk+t)^2)$  \\
\hline
\end{tabular}
}
\caption{Results based on a $2$ round algorithms. $T$ denotes the runtime to compute $1$-median/mean of a node distribution\protect\footnotemark, $B$ the information encoding a point and $I$ the information encoding a node in the uncertain data case.  $\Delta$ is the ratio between the maximum pairwise distance and the minimum pairwise distance in the dataset.}
\label{tab:main-results}
\end{table*}

\medskip\noindent
{\bf Our Results.}
We present our main results in Table~\ref{tab:main-results} and only present the results based on $2$ rounds. The full set of our results can be found in Appendix~\ref{sec:full-results}.  We use $T$ to denote the runtime to compute $1$-median/means of a node distribution, $B$ to denote the information needed to encode a point, and $I$ to denote the information needed to encode a node in the uncertain data case.  In the column of {\em Local Time}, the first is the local computation time of all sites, and the second is the local computation time at the coordinator. Observe that the total running time is $\tilde{O}(\sum_i n^2_i)$, which becomes $\tilde{O}(n^2/s)$ if the partitions are balanced.  This shows that we can reduce the running time by distributing the clustering across many sites.

In particular we have obtained the following. All algorithms finish in $2$ rounds in the coordinator model. We say a solution is an $(\alpha, \beta)$-approximation if it is a solution of cost $\alpha\cC$ while excluding $\beta t$ points, where $\cC$ is the optimum cost for excluding $t$ points.
\begin{enumerate}[(1)]
\item  We give $(O(1), 1)$-approximation algorithms with $\tilde{O}((sk + t)B)$ communication for the $(k,t)$-median (Section~\ref{sec:dist5}) and the $(k,t)$-center (Theorem~\ref{thm:center}) problems. The lower bounds in \cite{CSWZ16} for the $t=0$ case indicate that these communication costs are tight, if we want to output all the outliers (which our algorithms do), up to logarithmic factors. We also give an $(O(1+1/\eps), 1+\eps)$-approximation algorithm with $\tilde{O}((sk + t)B)$ communication for the $(k,t)$-median (with better running time) and the $(k,t)$-means (Theorem~\ref{thm:2-round}) problems. 

\item  We show that for $(k,t)$-median/means and $(k,t)$-center-pp the above results are achievable even on uncertain data (Theorem~\ref{cool-theorem}).  For uncertain $(k,t)$-center-g we obtain an $(O(1+1/\eps), 1+\eps)$-approximation algorithm with $\tilde{O}(skB + tI + s\log \Delta)$ communication, where $I$ is the information to encode the distribution of an uncertain point, and $\Delta$ is the ratio between the maximum pairwise distance and the minimum pairwise distance in the dataset (Theorem~\ref{thm:center-g}).
\end{enumerate}
Our results for the $(k,t)$-center problem improves that in \cite{Malkomes}. And as far as we are concerned, our results on distributed $(k, t)$-median/means and of uncertain input are the first of their kinds.
Our results for distributed $(k,t)$-median or means also lead to {\em subquadratic} time constant factor approximation centralized algorithms, which have been left open for many years.

\medskip\noindent{\bf Technical Overview.} The high level idea of our algorithms is fairly natural: Each site first performs a \emph{preclustering}, i.e., it computes some local solution on its own dataset. Then each site sends the centers of the local solution, number of attached points to each center and the ignored points to the coordinator, who will then solve the induced {\em weighted} clustering problem. 

A major difficulty is to determine how many points to ignore in the local solution at each site. Certainly for the sake of safety each site can ignore $t$ points and send all ignored $t$ points to the coordinator for a final decision. This would however incur $\Theta(st)$ bits of communication. To reduce the communication of this part to $O(t)$, we hope to find $\{t_1, \ldots, t_s\}$ such that $\sum_i t_i = t$ and each site $i$ sends a solution with just $t_i$ ignored points. At the cost of an extra round of communication, we solve the minimization problem $\sum_i f_i(t_i)$ subject to $\sum_i t_i = t$ for convex functions $\{f_i\}$. It is tempting to take $f_i(t_i)$ to be the cost of local solution with $t_i$ ignored points on site $i$, however, such $f_i$ is not necessarily convex. The remedy is to take a lower convex hull of $f_i$ instead, which can be shown to have only a mild effect on the solution cost. The convex hull of $t$ points can be found in $O(t\log t)$ time, and we can further reduce the runtime without compromising approximation ratio by computing local solutions on each site for only $\log t$ geometrically increasing values of $t_i$.

For uncertain data, it is natural to reduce the clustering problems to the deterministic case.  To this end, we `collapse' each node $j$ to its optimal center in $\PP$. For instance, for the $(k,t)$-median problem, each node $j$ is `collapsed' to $y_j = \argmin_{y\in \PP} \E_{\sigma}[d(\sigma(j),y)]$, called the $1$-median of node $j$. It may be tempting to consider the clustering problem on the set of $1$-medians, but the `collapse' cost is lost, hence we construct a \textit{compressed graph} $G$ that allows us to keep track of the collapse costs. The graph looks like a clique with tentacles, see Figure~\ref{fig:compressed_graph}. The $1$-medians form a clique in $G$ with edge weight being the distance in the underlying metric space; for each $1$-median $y_j$, we add a tentacle (an edge) from $y_j$ to a new vertex $p_j$ with edge weight being the collapse cost $\E_{\sigma}[d(\sigma(j),y_j)]$.  We manage to show that the original clustering problem is equivalent, up to a constant factor in cost, to the clustering problem on the compressed graph where the facility vertices are $1$-medians $\{y_j\}$ and the demand vertices are $\{p_j\}$. Our previous framework for deterministic data is then applied to the compressed graph.
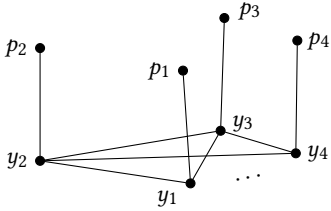
\begin{figure}
\centering
\begin{tikzpicture}
\draw (0,0) node[circle,fill,draw,scale=0.4](y1) {};
\draw (2,-0.3) node[circle,fill,draw,scale=0.4](y2) {};
\draw (2.4, 0.4) node[circle,fill,draw,scale=0.4](y3) {};
\draw (3.4, 0.1) node[circle,fill,draw,scale=0.4](y4) {};
\draw (y1)--(y2);
\draw (y1)--(y3);
\draw (y2)--(y3);
\draw (y3)--(y4);
\draw (y1)--(y4);
\node[rotate=10] at (2.8,-0.25) {$\cdots$};
\draw (0,1.5) node[circle,fill,draw,scale=0.4](p1) {};
\draw (1.9, 1.2) node[circle,fill,draw,scale=0.4](p2) {};
\draw (2.45, 1.9) node[circle,fill,draw,scale=0.4](p3) {};
\draw (3.42, 1.6) node[circle,fill,draw,scale=0.4](p4) {};
\foreach \i in {1,2,3,4}
	\draw (y\i)--(p\i);
\node at (-0.3, 0) {$y_2$};
\node at (-0.3, 1.5) {$p_2$};
\node at (1.7,-0.5) {$y_1$};
\node at (1.6,1.2) {$p_1$};
\node at (2.7,0.5) {$y_3$};
\node at (2.8,2) {$p_3$};
\node at (3.7,0.1) {$y_4$};
\node at (3.72,1.6) {$p_4$};
\end{tikzpicture}
\caption{An example of a compressed graph produced}
\label{fig:compressed_graph}
\end{figure}

Lastly, for the global center problem with uncertain data, we build upon the approach developed in \cite{GM09}, which uses a truncated distance function $\LL_\tau(x,y)=\max\{d(x,y)-\tau,0\}$ instead of the usual metric distance $d(\cdot,\cdot)$. Our algorithm performs a parametric search on $\tau$, and applies our previous framework to solve the global problem using local solutions. Now in the analysis of the approximation ratio we need to relate the optimum solution to the solution with truncated distance function, which is a fairly nontrivial task. 

\footnotetext{For a general discrete distribution on $m$ points in Euclidean space with $\PP$ be the whole space, $T = O(m)$ \cite{Dyer86}; for special distributions such as normal distribution, $T=O(1)$.}

\medskip\noindent\textbf{Related Work.}
In the centralized model, Charikar et al.\ gives a $3$-approximation algorithm for $(k,t)$-center, and an $(O(1), O(1))$ bicriteria algorithm for $(k,t)$-median~\cite{CKMN01}.  This bicriteria was later removed by Chen~\cite{Chen08}, who designed an $O(1)$-approximation algorithm using $\tilde{O}(k^2(k+t)^2n^3)$ time. Feldman and Schulman studied the $(k, t)$-median problem with different loss functions using the {\em coreset} technique~\cite{FS12}.

On uncertain data, Cormode and McGregor considered $k$-center/ median/means where each $\mathcal{D}_i$ is a discrete distribution~\cite{CM08}.  Guha and Munagala provided a technique to reduce the uncertain $k$-center to the deterministic $k$-median problem~\cite{GM09}. Wang and Zhang studied the special case of $k$-center on the line~\cite{WZ15}.  We refer the readers to the survey by Aggarwal~\cite{Aggarwal13}.

Clustering on distributed data has been studied only recently. In the coordinator model,  in the $d$-dimensional Euclidean space, Balcan et al.\ obtained $O(1)$-approximation algorithms with $\tilde{O}((kd + sk)B)$ bits of communication  for both $k$-median and $k$-means~\cite{BEL13}. Their results on $k$-means were further improved by Liang et al.~\cite{LBKW14} and Cohen et al.~\cite{CEMMP15}. Chen et al.\ provided a set lower bounds for these problems~\cite{CSWZ16}.  In the MapReduce model, Ene et al.\ designed several $O(1)$-approximation $O(1)$-round algorithms for the $k$-center and the $k$-median problems~\cite{EIM11}. Im and Moseley further studied the partial clustering variant~\cite{IM15}, however their algorithms require communication polynomial in $n$.  Cormode et al.\ studied the $k$-center maintenance problem in the distributed data stream model where the coordinator can keep track of the cluster centers at any time step~\cite{CMZ07}.


\section{Preliminaries}


\noindent{\bf Notation.}
We use the following notations in this paper.
\begin{itemize}
\item $\sol(Z,k,t,d)$: A solution (computed by an algorithm) to the median/means/center problem on point set $Z$ with at most $k$ centers and at most $t$ outliers, under the distance function $d$;
\item $\opt(Z,k,t,d)$: An optimal solution to the median/means
 or center problem on point set $Z$ with at most $k$ centers and at most $t$ outliers, under $d$;
\item $\Csol(Z,k,t,d)$: The cost of the solution $\sol(Z,k,t,d)$;
\item $\Copt(Z,k,t,d)$: The cost of the solution $\opt(Z,k,t,d)$;
\item $\pi(j)$: The center to which point $j$ is attached.
\end{itemize}
When $Z$ lies in a metric space and $d$ agrees with the distance function on the metric space, we omit the parameter $d$ in the notations above.

\smallskip\noindent{\bf Combining Preclustering Solutions.}
We review a theorem from \cite{GMMMO03}, which
concerns `combining' local solutions into a global solution. The problems considered in the theorem have {\em no} outliers ($t = 0$) and lie in a metric space, so we abbreviate the notation $\sol(Z,k,t,d)$ to $\sol(Z,k)$, etc.

\begin{theorem}[\cite{GMMMO03}]
Suppose that $\A = \A_1 \uplus \cdots \uplus\A_s$ (disjoint union) and $\{\sol(\A_i,k)\}$ are the preclustering solutions at sites. Let $\M = \{\pi(j): j\in \A\}$ and $L = \sum_{j\in \A} d(j, \pi(j))$, where $\pi(j)$ denotes the preclustering assignment. Consider the weighted $k$-median problem on $\M$ where the weight of $m \in \M$ is defined to be the number of points that are assigned to $m$ in the preclustering, that is, $\abs{\{j\ |\ j \in \A, \pi(j) = m\}}$. Then
\begin{enumerate}[(i)]\parskip=0pt
\item There exists a weighted $k$-median solution $\sol(\M,k)$ such that $\Csol(\M,k)\leq 2(L+\Copt(\A,k))$.
\item Given any weighted $k$-median solution $\sol(\M,k)$, there exists a $k$-median solution $\sol(\A,k)$ such that $\Csol(\A,k)\leq \sol(\M,k) + L$.
\end{enumerate}
Consequently, there exists a $k$-median solution $\sol(\A,k)$ such that $\Csol(\A,k)\leq 2\gamma(L+\Copt(\A,k)) + L$ and centers are restricted to $\M$, where $\gamma$ is the best approximation ratio for the $k$-median problem.
\label{thm:p2}
\end{theorem}

\begin{corollary}
\label{cor-compose}
The result in Theorem~\ref{thm:p2} extends to 
\begin{enumerate}[(i)]
\item the $k$-center problem;
\item the $k$-means problem with weaker constants, using a relaxed triangle inequality;
\item the $(k,t)$-median/means/center approximation on the 
weighted point set $\M$ (with $\gamma$ being the corresponding bicriteria approximation ratio),  {\em provided} the preclustering does not ignore any points. Otherwise the total number of ignored points is the sum of the ignored points in the clustering and preclustering phases.
\end{enumerate}
\end{corollary}

\section{$\protect\kt$-Median and $\protect\kt$-Means}
\label{sec:dist5}

Our algorithm for distributed $(k,t)$-median clustering is provided in Algorithm~\ref{alg:involved2a}. For integer pairs $(i,q)$, we consider the lexicographical order as partial order, that is, 
\begin{equation}\label{eqn:prec}
(i_1,q_1)\prec(i_2,q_2) \quad\text{if}\quad 
\left\{
\begin{array}{l}
i_1<i_2;\text{ or}\\
i_1=i_2\text{ and }q_1<q_2.
\end{array}
\right.
\end{equation}

\begin{algorithm}[ht]
\begin{algorithmic}[1]
\Input $\A = \A_1 \uplus \cdots \uplus \A_s$, $k\geq 1$, $t\geq 0$ and $\rho>1$
\Output $\sol(\A,k,(1+\epsilon)t)$ such that $\Csol(\A,k,(1+\epsilon)t) = O(1+1/\epsilon)\cdot\Copt(\A,k,t)$
\For{each site $i$}
\State $\I\gets \{ \lfloor \rho^r \rfloor: 1 \leq r \leq \lfloor \log_\rho t\rfloor, \ r\in \mathbb{\Z} \} \cup \{0, t\}$
\State Compute $\sol(\A_i,2k,q)$ for each $q\in\I$ \label{initial-step}
\State Compute the (lower) convex hull of the point set $\{(q,\Csol(\A_i,2k,q))\}_{q\in \I}$, which induces a function $f_i(\cdot)$ defined on $\{0,\dots,t\}$
\State Send the function $f_i(\cdot)$ to the coordinator
\EndFor
\State Coordinator computes $\ell(i,q)=f_i(q-1)-f_i(q)$ for each $1\leq i\leq s$ and each $1\leq q\leq t$
\State Coordinator \emph{stably} sorts all $\{\ell(i,q)\}$ in decreasing order\footnotemark
\State Coordinator finds $\ell(i_0,q_0)$ of rank\footnotemark $\rho t$ and sends $\ell(i_0, q_0)$, $i_0$ and $q_0$ to all sites \label{labelmerge}
\For{each site $i$}
\State $t_i \gets \max\{q: \ell(i,q) \geq \ell(i_0,q_0)\}$ \label{alg:step:t_i}
\Comment{define $\max\emptyset = 0$}
\If{$i = i_0$}\label{alg:step:if_block_begins}
\State $t_i\!\gets\! \min\{q \in \I: q \geq q_0\text{ and }\Csol(\A_i,2k,q_0) = f_{i_0}(q_0)\}$\label{alg:step-recomp}
\EndIf
\State Send the coordinator the $2k$ centers built in $\sol(\A_i,2k,t_i)$, the number of points
attached to each center, and the $t_i$ unassigned points \label{alg:mstep2}
\EndFor
\State Coordinator considers the union of the centers obtained from each site and the unassigned points, and  
applies Theorem~\ref{thm:alt-median} and outputs $\sol(\A,k,(1+\epsilon)t)$. 
\label{alg:finalstep}
\end{algorithmic}
\caption{Distributed $(k,(1+\epsilon)t)$-median clustering
\label{alg:involved2a}}
\end{algorithm}

\begin{remark} In Line~\ref{alg:finalstep} of Algorithm~\ref{alg:involved2a}, (i) no input point is ignored in the preclustering; (ii) if the preclustering aggregated $q$ points but the coordinator's algorithm chooses less than $q$ copies (to exclude exactly $t$) then the proofs are not affected in any way.
\end{remark}

\noindent
We begin with a theorem about approximating $(k,t)$-median or means with a different trade-off from that in \cite{CKMN01}. 
\begin{theorem}[Proof Omitted]
\label{thm:alt-median}\label{THM:ALT-MEDIAN}
Let $\epsilon > 0$.
We can compute $\sol(Z,k,(1+\epsilon)t)$ and $\sol(Z,(1+\epsilon)k,t)$ for the $(k,t)$-median problem in $\tilde{O}(|Z|^2)$ time such that
\begin{gather*}
\Csol(Z,k,(1+\epsilon)t)\leq \max\{6,6/\epsilon\}\cdot\Copt(Z,k,t), \text{ and }\\
\Csol(Z,(1+\epsilon)k,t)\leq \max\{6,6/\epsilon\}\cdot\Copt(Z,k,t).
\end{gather*}
The result extends to the $(k,t)$-means problem with a slightly larger constant.
\end{theorem}

Throughout the rest of the section, we denote by $t_i^\ast$ the number of ignored points from $\A_i$ in the global optimum solution $\opt(\A,k,t)$. 
We need the following lemmas.
 
\footnotetext{\textit{Stably} means that when $\ell(i_1,q_1) = \ell(i_2,q_2)$, the sorting algorithm puts $\ell(i_1,q_1)$ before $\ell(i_2,q_2)$ if $(i_1,q_1)\prec(i_2,q_2)$ as defined in \eqref{eqn:prec}.}
\footnotetext{Element of \textit{rank} $r$ means the $r$-th element in a sorted list}

\begin{lemma}
\label{lem1}
It holds that $\sum_i \Copt(\A_i,k,t^*_i)\leq 2\Copt(\A,k,t)$. For $(k,t)$-means the constant changes from $2$ to $4$.
\end{lemma}
\begin{proof}
We shall use an argument used in \cite{GMMMO03}. Let $\pi_{\textrm{opt}}$ be the center projection function and $K$ be the set of optimum centers in the optimal solution $\opt(\A,k,t)$. For each $\A_i$, we construct a solution $\sol(\A_i,k,t_i^\ast)$ by excluding the points excluded in $\opt(\A,k,t)$ and choosing $\left\{\argmin_{u\in\A_i} d(u, k): k\in K\right\}$ to be the centers. Then
\[
\Csol(\A_i,k,t_i^\ast)\leq 2\sum_{x\in \A_i} d(x,\pi_{\textrm{opt}}(x)).
\]
Summing over $i$ yields $\sum \Csol(\A_i,k,t_i^\ast)\leq 2\Copt(\A,k,t)$. The result for $k$-means follows from applying triangle inequality with $(a+b)^2\leq 2(a^2+b^2)$.
\end{proof}

\begin{lemma}\label{lem:waterfilling}
The $t_1,\dots,t_s$ computed in Step~\ref{alg:step:t_i} of Algorithm~\ref{alg:involved2a} minimizes $\sum_i f_i(t_i)$ subject to $\sum_i t_i \leq \rho t$ and $0\leq t_i\leq t$.
\end{lemma}
\begin{proof}
Suppose that $t_1',\dots,t_s'$ is a minimizer. Since $f_i(\cdot)$ is non-increasing for all $i$, it must hold that $\sum_i t_i' = \rho t$. By the definition of $t_i$, it also holds that $\sum_i t_i = \rho t$. If $(t_1',\dots,t_s')\neq(t_1,\dots,t_s)$, there must exist $i, j$ such that $t_i'>t_i$ and $t_j'<t_j$. By the definition of $t_i$ and the sorting of $\{\ell(i,q)\}$, we know that
\[
\ell(i,t_i+1) \leq \ell(i_0,q_0),\quad
\ell(j,t_j) \geq \ell(i_0,q_0).
\]
From convexity of $f_i$ and that $t_i'\geq t_i+1$ and $t_j' + 1\leq t_j$, it follows that
\[
f_i(t_i'-1) - f_i(t_i') \leq \ell(i_0,q_0)\leq f_{j}(t_j') - f_{j}(t_j'+1)
\]
which means that increasing $t'_{j}$ by $1$ and decreasing $t'_{i}$ by $1$ will not decrease the sum
\[
G(q_1',\dots,q_s') :=\sum_i (f_i(0) - f_i(t_i')).
\]
Therefore $\sum_i f_i(t_i') = \sum_i f_i(0) - G(t_1',\dots,t_s')$ will not increase. We can continue this procedure until $(t_1',\dots,t_s') = (t_1,\dots,t_s)$.
\end{proof}

\begin{lemma}\label{lem:exception} It holds for all $i\neq i_0$ that $t_i\in\I$ and $\Csol(\A_i,2k,t_i) = f_i(t_i)$, where $i_0$ is computed in Step~\ref{labelmerge} and $t_i$'s in Step~\ref{alg:step:t_i} of Algorithm~\ref{alg:involved2a}.
\end{lemma}
\begin{proof}
Since $0\in \I$, we need only to consider the $i$'s with $t_i\neq 0$.  By the selection of $i_0$ and $q_0$, it must hold that
\begin{gather*}
\ell(i,t_i)\geq \ell(i_0,q_0)> \ell(i,t_i+1) \quad \text{for} \quad i < i_0\\
\ell(i,t_i)> \ell(i_0,q_0)\geq \ell(i,t_i+1) \quad \text{for} \quad i > i_0,
\end{gather*}
which implies that $\ell(i,t_i) > \ell(i,t_i+1)$ whenever $i\neq i_0$, i.e.,
\[
f_i(t_i-1)-f_i(t_i) > f_i(t_i)-f_i(t_i+1),\quad i\neq i_0.
\]
Hence $(i,f_i(t_i))$ is a vertex of the convex hull for all $i\neq i_0$, that is, $t_i\in\I$ and $f_i(t_i) = \Csol(\A_i,2k,t_i)$.
\end{proof}
Now we are ready to bound the `goodness' of local solutions.

\begin{lemma}
\label{lem:goodness} Let $\rho=2$. It holds that
$\sum_i \Csol(\A_i,2k,t_i) \leq 12\cdot \Copt(\A,2k,t)$ and $\sum_i t_i \leq 3t$, where $t_1,\dots,t_s$ are computed in Step~\ref{alg:step:t_i} and may be updated in Step~\ref{alg:mstep2} of Algorithm~\ref{alg:involved2a}.
\end{lemma}
\begin{proof}
Let $\hat t_i = \min\{q \in \I: q \geq t_i^\ast\}$.
It follows from Lemma~\ref{lem1} with 
$\sum_i t^*_i \leq t$ that
\begin{multline*}
2\Copt(\A,k,t) \geq \sum_i \Copt(\A_i,k,t^*_i) \geq \sum_i \Copt(\A_i,k,\hat{t}_i) \\ \geq \frac1{6} \sum_i \Csol(\A_i,2k,\hat{t}_i),
\end{multline*}
where the last inequality follows from Theorem~\ref{thm:alt-median} (applied with $\eps = \rho - 1 = 1$). Observe that $\hat{t}_i\leq 2t_i^\ast$ and thus $\sum_i \hat{t}_i\leq 2\sum_i t_i^\ast \leq 2t$, and
\[
\sum_i \Csol(\A_i,2k,\hat{t}_i) \geq \sum_i f_i (\hat{t}_i) \geq \sum_i f_i(t_i),
\]
where the last equality follows from Lemma~\ref{lem:waterfilling}, and $t_i$'s are computed in Step~\ref{alg:step:t_i}.

Now, by Lemma~\ref{lem:exception}, $f_i(t_i) = \Csol(\A_i,2k,t_i)$ for all except one $i$. The exceptional $t_i$ will be replaced by a bigger value, which will not increase $f_i(t_i)$ by the monotonicity of $f_i$, and the first part follows. This update will increase $\sum_i t_i$ by at most $t$ and thus $\sum_i t_i \leq 3t$.
\end{proof}

Lemma~\ref{lem:goodness} and Theorem~\ref{thm:alt-median} together give the following. Note that $|\I|=O(\log t)$.

\begin{theorem}
\label{thm:2-round}
For the distributed $(k,t)$-median problem, Algorithm~\ref{alg:involved2a} with $\rho=2$ outputs $\sol(\A,k,(1+\epsilon)t)$ satisfying $\Csol(\A,k,(1+\epsilon) t))\leq O(1+1/\epsilon)\cdot \Copt(\A,k,t)$. The sites communicate a total of $\tilde{O}(sk + t)$ bits of information with the coordinator over $2$ rounds. The runtime at each site is $\tilde O(n_i^2)$ and the runtime at the coordinator is $\tilde O((sk+t)^2)$. The same result holds for $(k,t)$-means with larger constants in the approximation ratio and the runtime.
\end{theorem}
\begin{proof}
The communication cost is straightforward. By Lemma~\ref{lem:goodness}, the coordinator will solve the problem of at most $2sk+3t$ points. The claims on approximation ratio and the runtime then follow from Theorem~\ref{thm:alt-median}, noting that it takes time $O(\I \log\I) = \tilde O(1)$ to find the convex hull. 
\end{proof}
If we were only interested in the clustering and not the list of ignored points, we could set $\rho = 1+\delta$ and change line~\ref{alg:step:if_block_begins} to line~\ref{alg:mstep2} of Algorithm~\ref{alg:involved2a} to the following. The sites do not send the ignored nodes but just the number of them, and the exceptional site runs a slightly more convoluted algorithm.
\medskip
\begin{algorithmic}[1]
\makeatletter
\setcounter{ALG@line}{11}
\makeatother
\If{$i\neq i_0$}
\State Send the coordinator $t_i$, the $2k$ centers built in $\sol(\A_i,2k,t_i)$ and the number of points
attached to each center
\Else
\State $t_{i,1} = \max\{q \in \I: q \leq t_i\text{ and }\Csol(\A_i,2k,q) = f_i(q)\}$
\State $t_{i,2} = \min\{q \in \I: q \geq t_i\text{ and }\Csol(\A_i,2k,q) = f_i(q)\}$
\State Combine $\sol(\A_i,2k,t_{i,1})$ and $\sol(\A_i,2k,t_{i,2})$ to form a solution $\sol(\A_i,4k,t_i)$ by taking the union of the medians, attaching each point to the closest center among the combined centers, and ignoring the points with largest $t_i$ distances.
\State Send to the coordinator $t_i$, the combined centers and the number of points attached to each center.
\EndIf
\end{algorithmic}
\medskip
Observe that Lemma~\ref{lem:goodness} still holds with $\sum_i t_i \leq (1+\delta)t$, since we are not changing the exceptional $t_i$. For the exceptional site $i$, suppose that $t_i = (1-\theta)t_{i,1} + \theta t_{i,2}$ for some $\theta\in (0,1)$, we have $(1-\theta)f_i(t_{i,1}) + \theta f_i(t_{i,2}) \leq f_i(t_i)$.
We now argue the next critical lemma.

\begin{lemma}
\label{lem:convex}
$ \Csol(\A_i,4k,t_i)\leq (1-\theta)f_i(t_{i,1}) + \theta f_i(t_{i,2})$.\end{lemma}
\begin{proof}
We will prove the lemma by carefully designing an assignment of $n-t_i$ points to the $4k$ centers which is bounded above by the right hand side. Since choosing the minimum $n-t_i$ distances will only result in a smaller value, the lemma would follow.

For $j=1,2$, let $\pi_j$ be the center projection function in $\sol(\A_i,2k,t_{i,j})$ and $P_i$ the set of clustered points in $\sol(\A_i,2k,t_{i,j})$. For $x\in P_1\cap P_2$, we attach $x$ to the nearer one between the two centers $\pi_1(x)$ and $\pi_2(x)$, and the incurred cost is 
\begin{equation}\label{eqn:cost1}
\min\{d(x,\pi_1(x)),d(x,\pi_2(x))\} 
\leq (1-\theta)d(x,\pi_1(x)) + \theta d(x,\pi_2(x)).
\end{equation}

For $x\in P_1\triangle P_2$, since only one of $\pi_1(x)$ and $\pi_2(x)$ exist, we abbreviate it as $\pi(x)$ for simplicity. Define $h(x)$ for each $x\in P_1\triangle P_2$ as
\[
h(x) = \begin{cases}
	(1-\theta)\cdot d(x,\pi(x)), & x\in P_1\setminus P_2;\\
	\theta\cdot d(x,\pi(x)), & x\in P_2\setminus P_1.
\end{cases}
\]

Let $r = |P_1\cap P_2|$, $r_1 = |P_1\setminus P_2|$ and $r_2 = |P_2\setminus P_1|$. It holds that 
$r+r_1= n -t_{i,1}$ and $r+r_2=n-t_{i,2}$, thus $r_1 > r_2$ and
\[
(1-\theta)r_1 + \theta r_2 = n - t_i - r.
\]
Define $Q_1 = P_1\setminus P_2$ and $Q_2 = P_2\setminus P_1$. Pick $x = \argmin_{z\in Q_1\cup Q_2} h(z)$. If $x\in Q_1$, pick an arbitrary $u\in Q_2$, otherwise pick $u \in Q_1$. Attach $x$ to $\pi(x)$ in the $4k$-center solution we are constructing and mark $u$ as outlier. Note that this incurs a cost of 
\begin{equation}\label{eqn:cost2}
d(x,\pi(x))\leq 
\begin{cases}
(1-\theta)d(x,\pi(x)) + \theta d(u, \pi(u)),\quad x\in Q_1;\\
(1-\theta)d(u,\pi(u)) + \theta d(x, \pi(x)),\quad x\in Q_2,\\
\end{cases}
\end{equation}
by our choice of $x$, because one of the combination terms is exactly $h(x)$ and it is smaller than $h(u)$, which is exactly the other term. Then we remove $x$ and $u$ from $Q_1$ or $Q_2$ depending on the case. Now, $|Q_1| = r_1 - 1$ and $|Q_2| = r_2 - 1$, and note that
\[
(1-\theta)(r_1 - 1) + \theta (r_2 - 1) = n - t_i - r - 1.
\]

Since $r_1 > r_2$, we can continue this process until $Q_2 = \emptyset$. At this point we have run the procedure above $r_2$ times, and it holds that
\[
(1-\theta) r_1 = n - t_i - r - r_2.
\]

Note that $r_1\geq n - t_i - r - r_2$, so we can choose $E\subseteq Q_1$ to be the points with smallest $n-t_i-r-r_2$ values of $h$. Attach points in $E$ to their respective centers and mark the remaining points in $Q_1$ as outliers. This incurs a cost of
\begin{multline}
\label{eqn:cost3}
\sum_{x\in E} d(x,\pi(x))\leq \frac{n-t_i-r-r_2}{r_1}\sum_{x\in Q_1} d(x,\pi(x)) \\= (1-\theta)\sum_{x\in Q_1} d(x,\pi(x))
\end{multline}
In total we have assigned $r + r_2 + (n-t_i-r-r_2) = n - t_i$ points as desired.  The desired upper bound on cost follows from (i) summing both sides of \eqref{eqn:cost1} over $P_1\cap P_2$; (ii) summing both sides of \eqref{eqn:cost2} over $x$ and the corresponding $u$ during the pairing procedure; and (iii) Equation~\eqref{eqn:cost3}. Note that (ii) covers $(P_1\triangle P_2)\setminus Q_1$, where $Q_1$ is the post-pairing set.
\end{proof}

As a consequence of Lemma~\ref{lem:convex}, $
\Csol(\A_i,4k,t_i)\leq f_i(t_i)$. Thus the upper bound on the approximation ratio still holds.
Finally, note that $|\I| = \tilde{O}(1/\delta)$ and we conclude that

\begin{theorem}
\label{thm:clustering_only}
For the distributed $(k,t)$-median problem, the modified Algorithm~\ref{alg:involved2a} with $\rho=1+\delta$ outputs $\sol(\A,k,(2+\epsilon+\delta)t)$ satisfying $\Csol(\A,k,(2+\epsilon+\delta)t)\leq O(1+1/\epsilon)\cdot\Copt(\A,k,t)$. The sites communicate a total of $\tilde O(s\delta^{-1} + skB)$ bits of information with the coordinator over $2$ rounds. The runtime on site $i$ is $\tilde{O}(n_i^2/\delta)$ and the runtime on the coordinator is $\tilde O((sk)^2)$. The same result holds for $(k,t)$-means with a larger constant in the approximation ratio.
\end{theorem}

\subsection{Subquadratic-time Centralized Algorithm}
We now show an unusual application of Theorem~\ref{thm:2-round} in speeding up existing constant-factor approximation algorithms for $(k,t)$-median (or means). Note that the centralized bicriteria approximation algorithms in Charikar~\cite{CKMN01} are 
$\tilde{O}(n^3)$ from $n$ points, and while the modifications in Theorem~\ref{thm:alt-median} improve the running time to 
$\tilde{O}(n^2)$, this leaves open the important question: {\em Are there algorithms with provable constant factor approximation guarantees which are subquadratic?} Observe that the question is even more pertinent in the context of unicriterion approximation, for which the only known result is a $\tilde{O}(n^3k^2t^2)$-time constant-factor approximation of $(k,t)$-median \cite{Chen08}. In the sequel we show that the running time can be brought to almost linear time. The improvement arises from the fact that we can simulate a distributed algorithm sequentially.

\begin{lemma}
\label{recurselemma}
Suppose that we are given a $\tilde{O}(n^{1+\alpha_0}k^2)$ time algorithm for bicriteria approximation which produces $2k$ centers or $2t$ outliers with approximation factor $\gamma$, where $\alpha_0 \leq 1$. Then we can produce a similar algorithm with running time $\tilde{O}(t^2)+\tilde{O}\left(n^{\frac{2+2\alpha_0}{2+\alpha_0}}k^{2}\right)$ and approximation $c_0\gamma$ for some absolute constant $c_0 > 0$.
\end{lemma} 
\begin{proof}
We will apply Theorem~\ref{thm:2-round} after dividing the data arbitrarily in $s$ pieces of size $n/s$. The sequential simulation of the $s$ sites will take time $\tilde{O}(s \left(n/s\right)^{1+\alpha_0}k^2)$ based on the statement of the lemma. The coordinator will require time $\tilde{O}((sk+t)^2)=\tilde{O}(s^2k^2) + \tilde{O}(t^2)$. Observe that we can now balance $n^{1+\alpha_0} = s^{2+\alpha_0}$, which provides us the optimum $s$ to use and achieve a running time of 
\[
\tilde{O}(t^2)+\tilde{O}(s^2k^2) = \tilde{O}(t^2)+\tilde{O}\left(n^{\frac{2+2\alpha_0}{2+\alpha_0}}k^{2}\right).\qedhere
\]
\end{proof}

\begin{theorem} Let $\alpha > 0$ and suppose that $t\leq \sqrt{n}$. There exists a centralized algorithm for the $(k,t)$-median problem that runs in $\tilde O(n^{1+\alpha}k^2)$ time and outputs a solution $\sol(\A,k,2t)$ satisfying $\Csol(\A,k,2t)\leq (1+1/\alpha)^{O(1)}\Copt(\A,k,t)$.
\end{theorem}
\begin{proof}
Note that the algorithm in Theorem~\ref{thm:alt-median} has runtime $\tilde O(n^2)$, so we can take $\alpha_0=1$ in Lemma~\ref{recurselemma} to obtain an algorithm of approximation ratio $\gamma = 6$ and runtime $\tilde O(t^2 + n^{4/3}k^2)$, which is $\tilde O(n^{4/3}k^2)$ by our assumption that $t\leq \sqrt{n}$. Repeatedly applying Lemma~\ref{recurselemma} for $j$ times gives an algorithm of runtime $\tilde O(n^{1+1/(2^j-1)}k^2)$ and approximation ratio $(c_0\gamma)^j$. Let $j = \log(1+1/\alpha)$, the runtime becomes $O(n^{1+\alpha}k^2)$ and the approximation ratio $(1+1/\alpha)^{\log(c_0\gamma)} = (1+1/\alpha)^{O(1)}$.
\end{proof}
\begin{remark} We remark that
\begin{enumerate}[(i)]\parskip=0pt
	\item the theorem above also holds for $\sol(\A,2k,t)$, where the number of centers,  instead of the outliers, is relaxed.
	\item for the unicriterion approximation, if we use the algorithm of runtime $\tilde O(n^3t^2k^2)$ from \cite{Chen08} instead of the result of Theorem~\ref{thm:alt-median}, we need to balance $s^3$ and $s(n/s)^{1+\alpha_0}$ for an analogy of Lemma~\ref{recurselemma}, which will eventually lead to an algorithm of runtime $O(n^{1+\alpha}t^2k^2)$, provided that $t\leq n^{1/5}$.
\end{enumerate}
\end{remark}

\section{$\protect\kt$-Center Clustering}
\label{sec:center1}

Our algorithm for $(k,t)$-center clustering is presented in Algorithm~\ref{alg:involved2b}. It is similar to Algorithm~\ref{alg:involved2a} but only simpler, because the preclustering stage admits a simpler algorithm due to Gonzalez~\cite{Gonzalez85}. For the $k$-center problem on a point set $Z$ of $n$ points, Gonzalez's algorithm outputs a re-ordering of points in $Z$, say, $p_1,\dots,p_n$, such that for each $1\leq r\leq n$, the solution $\sol(Z,r)$ of choosing $\{p_1,\dots,p_r\}$ as the $r$ centers is a $2$-approximation for the $r$-center problem on $Z$, i.e., $\Csol(Z,r)\leq 2\Copt(Z,r)$.

\begin{algorithm}[htbp]
\begin{algorithmic}[1]
\For {each site $i$}
\State Run Gonzalez's algorithm and obtain a re-ordering $\{a_1,\dots,a_{n_i}\}$ of the points in $\A_i$
\For {each $1\leq q\leq t$}
	\State Compute $\ell(i,q)\gets\min\{d(a_j,a_{k+q}): j < k+q\}$
\EndFor
\EndFor
\State Sites and coordinator sort $\{\ell(i,q)\}$, and follow the subsequent steps as in Algorithm~\ref{alg:involved2a}, where the coordinator in the last step runs the  algorithm in \cite{CKMN01} for the $k$-center problem with exactly $t$ outliers.
\end{algorithmic}
\caption{Distributed $(k,t)$-center clustering
\label{alg:center}\label{alg:involved2b}}
\end{algorithm}
The core
argument is that the $k$-center algorithm of Gonzalez 
can be used to simultaneously (a) precluster the local
data into local solutions and (b) provide a witness that can be compared globally. 

\begin{remark} In Algorithm~\ref{alg:center}, (i) none of the original points is ignored in the preclustering, and (ii)
it is possible that the preclustering aggregated $q$ points but the coordinator's algorithm chooses less than $q$ copies to exclude exactly $t$ points. This does not affect the proofs of 
$(k,t)$-center clustering.
\end{remark}

We now analyze the performance of Algorithm~\ref{alg:center}. Denote by $t_i^\ast$ the number of points ignored from $\A_i$ in the  global optimum  solution $\opt(\A,k,t)$. First we show two structural lemmas. 

\begin{lemma} 
\label{lem1-center}
$2\Copt(\A_i,k,t) \geq \displaystyle \max_i \Copt(\A_i,k,t^*_i)$.
\end{lemma}
\begin{proof}
Use the same argument in the proof of Lemma~\ref{lem1}.
\end{proof}

\begin{lemma}
\label{lem2-center}$ \max\limits_i \Copt(\A_i, k, t^*_i) \geq \displaystyle \min_{\sum_i t_i \geq t} \left( \max_i \Copt(\A_i, k, t_i) \right).$
\end{lemma}
\begin{proof}
It follows from the fact that $\sum_i t_i^\ast = t$.
\end{proof}

\begin{theorem}
For the distributed $(k,t)$-center problem, Algorithm~\ref{alg:center} outputs $\sol(\A,k,t)$ satisfying $\Csol(\A,k,t)\leq O(1)\cdot\Copt(\A,k,t)$. The sites communicate a total of $\tilde{O}((sk + t)B)$ bits of information to the coordinator over $2$ rounds. The runtime on site $i$ is $\tilde O((k+t)n_i)$ and the runtime on the coordinator is $\tilde O((sk+t)^2)$.
\label{thm:center}
\end{theorem}
\begin{proof}
The approximation ratio follows from a similar argument to that of Theorem~\ref{thm:2-round}, using Lemma~\ref{lem1-center} and~\ref{lem2-center}. The coordinator runtime follows from  \cite[Theorem~3.1]{CKMN01} and the site runtime from~\cite{Gonzalez85}, noting that we need only the first $k+t$ points of the reordering of each $\A_i$. The communication cost is clear from Algorithm~\ref{alg:involved2b}.
\end{proof}



\section{Clustering Uncertain Input}
\label{sec:uncertain}

Recall that in the setting of clustering with uncertainty there is an underlying metric space $(\PP,d)$. We are given a set of input nodes $j
\in \A$ which correspond to distributions $\D_j$ on $\PP$. In this 
section we shall use nodes to indicate the input and points to indicate deterministic objects 
in the metric space $\PP$.  We shall denote by $\sigma(j)$ a realization of node $j$ and by $\pi(j)$ the center node to which $j$ is attached.
Our goal in the $(k,t)$-median problem in this context is to compute 

\begin{equation}
\min_{\stackrel{K \subseteq \PP, \O \subseteq \A}{ |K|\leq k,|\O|\leq t}} \left[ \sum_{j\in \A\setminus \O} \left( \min_{\pi(j)} \E_{\sigma} \left[d(\sigma(j),\pi(j))\right] \right) \right].
\label{eqn:median-un}
\end{equation}
For $(k,t)$-means we use $d^2(\cdot,\cdot)$ and for $(k,t)$-center-pp we use $\max_j$ instead of $\sum_j$. 

Define $\widehat{d}: \A\times \PP\to \R$ as
$
\widehat{d}(j, p) = \E_{\sigma}[d(\sigma(j),p)],
$
the objective function (\ref{eqn:median-un}) is then reduced to the usual $(k,t)$-median problem with
the new distance function $\widehat{d}$. However, this definition only allows the
computation of distance between an input node and a point in $\PP$. To extend $\widehat{d}$ to a pair of input nodes, the site
holding $\A_i$ will need to know the point set $\bigcup_{j\in \A_{i'}} \supp(\mathcal{D}_j)$ from some
other site $i'$. This will blow up the communication cost, and thus
naively using this distance function in combination with the algorithms
developed previously will not work well. To circumvent this issue we combine the notion of
$1$-median introduced in
\cite{CM08} along with the framework in Theorem~\ref{thm:p2}, and introduce a compression scheme to evaluate distances.

\begin{definition}
\label{onemeddef}
For each node $j$, define its $1$-median and $1$-mean to be
\[
y_j = \argmin_{y\in \PP} \E_{\sigma}[d(\sigma(j),y)], \quad y_j' = \argmin_{y\in \PP} \E_{\sigma}[d^2(\sigma(j),y)],
\]
respectively.
\end{definition}

\begin{definition}[Compressed graph] 
\label{def:compress}
The compressed graph $G(\A)$ is a weighted graph on vertices $\PP \cup \{p_j\}_{j\in \A}$, where the edges are as follows: (1) each pair $(u,v)\in \PP$ is an edge with weight $d(u,v)$, and (2) for each $j\in\A$, the vertex $p_j$ is  connected only to $y_j$ with weight $\ell_j = \E_{\sigma}[d(\sigma(j),y_j)]$.  Define the distance $d_G(u, v)$ between two vertices $u$, $v$ in $G$ to be the length of the shortest path between $u$ and $v$ in $G$.
\end{definition}

For the compressed graph $G$, we can also consider the following $(k,t)$-median problem, where we restrict the demand points to $\{p_j\}$ and the possible centers to $\{y_j\}$, and the distance function is the length of shortest path on $G$. We continue to use the notations $\sol(G,k,t)$, $\Csol(G,k,t)$, etc., to denote the solution and the corresponding cost of $(k,t)$-median problem on $G$. The following two lemmas show that $(k,t)$-median problem in Eqn~\eqref{eqn:median-un} is, up to some constant factor in the approximation ratio, equivalent to the $(k,t)$-median problem on the compressed graph.

\begin{lemma}\label{lem:A=>graph}
If there exists a solution $\sol(\A,k,t)$ of cost $\Csol(\A,k,t)$ to the objective in Equation~\eqref{eqn:median-un}, then there exists a solution $\sol(G(\A),k,t)$ on the compressed graph such that $\Csol(G(\A),k,t)\leq 5\Csol(\A,k,t)$.
\end{lemma}
\begin{proof}
Let $\A'$ be the set of clustered nodes in the feasible
$(k,t)$-median solution of the original problem with the objective in \eqref{eqn:median-un}. Define the set of center points $M = \{y_j: j\in \A'\}$. For each $j\in \A'$, let $y_{\pi(j)} = \argmin_{y\in M} d(\pi(j), y)$. Let $\sol(G(\A),k,t)$ be the solution of connecting each point $p_j$ ($j\in \A'$) to $y_{\pi(j)}$ in the compressed graph $G$. We try to upper bound the cost $\Csol(G(\A),k,t)$:
\begin{align}
& \Csol(G(\A),k,t) = \sum_{j \in \A'} d_G(y_{\pi(j)},p_j)  \tag*{(definition of $\Csol$)}\\
&= \sum_{j \in \A'} \left( d(y_{\pi(j)},y_j) +  d_G(y_j,p_j) \right) \quad \tag*{(definition of $d_G$)}\\
&\leq \sum_{j \in \A'} d(y_{\pi(j)},\pi(j)) + \sum_{j \in \A'} d(\pi(j),y_j) + \sum_{j \in \A'} d_G(y_j,p_j) \tag*{(triangle inequality)}\\
&\leq 2\sum_{j \in \A'} d(\pi(j),y_j) + \sum_{j\in \A'} \ell_j, \nonumber
\end{align}
where the last line follows from $d(y_{\pi(j)},\pi(j))\leq d(\pi(j),y_j)$ by the definition (optimality) of $y_{\pi(j)}$.

Observe that for any realization $\sigma(j)$, it holds that
\[
d(y_j,\pi(j)) \leq d(y_j,\sigma(j)) + d(\sigma(j),\pi(j)).
\]
Taking expectation over $\sigma$,
\[
d(y_j,\pi(j)) \leq \E_\sigma d(y_j,\sigma(j)) + \E_\sigma d(\sigma(j),\pi(j)) 
= \ell_j + \E_\sigma d(\sigma(j),\pi(j)).
\]
Summing over $j\in \A'$,
\begin{equation}\label{eqn:heart}
\sum_{j\in A'} d(y_j,\pi(j))
\leq \sum_{j\in \A'} \ell_j + \sum_{j\in \A'} \E_\sigma d(\sigma(j),\pi(j))
\leq \sum_{j\in \A'} \ell_j + \Copt(\A,k,t).
\end{equation}
%
We next bound $\sum_{j\in \A'} \ell_j$. This is exactly the cost of connecting each $j\in\A'$ to its $1$-median, which is the optimal solution of at most $n-t$ centers for $\A'$. The optimal cost for $n-t$ centers is clearly less than that for $k$ centers and hence $\sum_{j\in \A'} \ell_j\leq \Copt(\A,k,t)$.

Therefore $\Csol(G(\A),k,t)\leq 2\cdot 2\Copt(\A,k,t) + \Copt(\A,k,t) = 5\Copt(\A,k,t)$ as claimed. 
\end{proof}

\begin{lemma}\label{lem:graph=>A}
If there exists a solution $\sol(G(\A),k,t)$ of cost \\
$\Csol(G(\A),k,t)$ on the compressed graph, then there exists a solution $\sol(\A,k,t)$ for the problem formulated in \eqref{eqn:median-un} such that $\Csol(\A,k,t)\leq 2\Csol(G(\A),k,t)$.
\end{lemma}
\begin{proof}
Let $\A''$ be the set of clustered nodes in $\sol(G(\A),k,t)$. A similar argument of increasing the number of centers as in Lemma~\ref{lem:A=>graph} yields that $\sum_{j \in \A''} \ell_j \leq \Csol(G(\A),k,t)$. Suppose that $p_j$ is assigned to $\pi(j)$ in $\sol(G(\A),k,t)$ in the compressed graph. Note that $\pi(j)\in\PP$. Let $\sol(\A,k,t)$ be the solution of attaching $j$ to $\pi(j)$ in $\PP$, and the cost can be bounded as
\begin{align}
\Csol(\A,k,t) &= \sum_{j \in \A''}  \E_{\sigma} \left( d(\sigma(j),\pi(j)) \right) \tag*{(definition of $\Csol$)}\\
& \leq  \sum_{j \in \A''}  \E_{\sigma} \left( d(\sigma(j),y_j) \right) + \sum_{j \in \A''} d(y_j,\pi(j)) \tag*{(triangle inequality)} \\
& \leq  \sum_{j \in \A''} \ell_j + \sum_{j \in \A''} d_G(p_j,\pi(j)) \tag*{(definition of $d_G$, see below)}\\
&\leq 2\Csol(G(\A),k,t) \tag*{(definition of $\Csol$)},
\end{align}
where the third line follows from $d_G(p_j,\pi(j)) = d(p_j, y_j) + d(y_j,\pi(j)) \geq d(y_j,\pi(j))$.
\end{proof}

The equivalence between the original problem and the one on the compressed graph also holds for the $(k,t)$-center-pp and the $(k,t)$-means problems.

\begin{lemma}\label{lem:equiv_others}
Lemma~\ref{lem:A=>graph} and Lemma~\ref{lem:graph=>A} both hold 
\begin{enumerate}[(a)]\parskip=0pt
	\item for $(k,t)$-center-pp with the same constants; and
	\item for $(k,t)$-means with slightly larger constants. 
\end{enumerate}
\end{lemma}
\begin{proof}
\begin{enumerate}[(a)]\parskip=0pt
\item Observe that $\sum_j$ is replaced with $\max_{j}$ and Equation~\eqref{eqn:heart} rewrites to
\[ 
\max_{j \in \A'}  d(y_j,\pi(j)) \leq \max_{j \in \A'} \ell_j + \Copt(\A,k,t).
\]
The remainder of the equations hold with this transformation. 
\item Note that we used triangle inequality in the proof above. Although the square of the distance does not obey the triangle inequality, we can nevertheless apply $(a+b)^2 \leq 2a^2 + 2b^2$ after the triangle inequality. The derivations above will go through and the results hold with slightly larger constants.\qedhere
\end{enumerate}
\end{proof}

The overall algorithm is summarized in Algorithm~\ref{alg:uncertainty}.
Note that we cannot just cluster the $\{y_j\}$; 
the graph is necessary. To implement the algorithm, we need to show that each site is able to compute the distance function individually. Indeed, note that any site that contains $p_j$ will also
contain the corresponding $y_j$ or $y'_j$ and the value
$\E_{\sigma}[d(\sigma(j),y_j)]$ or $\E_{\sigma}[d^2(\sigma(j),y'_j)]$
respectively. Therefore the distance oracle on the graph can be implemented by the site in constant time.

\begin{algorithm}[t]
\begin{algorithmic}[1]
	\For {each site $i$}
		\State Compute $\ell_j = \E_{\sigma}[d(\sigma(j),y_j)]$ for all $j\in \A_i$\label{alg:uncertainty:1-median}
		\State Construct the compressed graph of $\A_i$ as described in Definition~\ref{def:compress}
		\State Run any algorithm corresponding to Section~\ref{sec:dist5} and Section~\ref{sec:center1} on the compressed graph, with the following change: whenever the site has to communicate $p_j$, it also sends $y_j$ (or $y'_j$) and the values of $\E_{\sigma}[d(\sigma(j),y_j)]$ (or $\E_{\sigma}[d^2(\sigma(j),y_j')]$).
\EndFor
\end{algorithmic}
\caption{A Compression Scheme for Distributed Partial Clustering of Uncertain Data
\label{alg:uncertainty}
\label{ALG:UNCERTAINTY}}
\end{algorithm}

\begin{theorem}
\label{cool-theorem}
For the distributed $(k,t)$-median problem, Algorithm~\ref{alg:uncertainty} outputs $\sol(\A,k,(1+\epsilon)t)$ such that $\Csol(\A,k,(1+\epsilon)t) = O(1+1/\epsilon)\cdot\Copt(\A,k,t)$. The sites communicate a total of $\tilde O((sk+t)B)$ bits of information to the coordinator over $2$ rounds. The runtime on site $i$ is $\tilde O(n_i^2 + n_iT)$, where $T$ is the runtime to compute 1-median, and the runtime on the coordinator is $\tilde O((sk+t)^2)$. The same result holds for the $(k,t)$-median and center-pp problems with larger constants.
\end{theorem}
\begin{proof}
By Lemma~\ref{lem:graph=>A} for the median problem and Lemma~\ref{lem:equiv_others} for the means and center-pp problems, it suffices to show that we can solve the $(k,t)$-median problem on the compressed graph. The result then follows from Theorem~\ref{thm:2-round} and Theorem~\ref{thm:clustering_only} with the following amendments: When a site sends the $t$ or $t_i$ potential outliers, it needs to send the $y_j$ and the corresponding values $\E_{\sigma}[d(\sigma(j),y_j)]$ or $\E_{\sigma}[d^2(\sigma(j),y_j')]$, which at most doubles the communication cost. The runtime is increased by $O(n_iT)$ due to Step~\ref{alg:uncertainty:1-median} since computing $\ell_j$ on the compressed graph takes $O(T)$ time.
\end{proof}

Other results claimed in Table~\ref{tab:our_results} follow from analogous amendments to Theorem~\ref{thm:clustering_only}.

\smallskip{\bf The global $k$-Center case.} We now focus on $(k,t)$-center-g. In this setting $\D_j$'s are independent and we optimize
\[
\min_{\stackrel{K \subseteq \PP, \O \subseteq \A}{|K|\leq k,|\O|\leq t}} \left( \E_{\sigma \sim \prod_j \D_j} \left[ \max_{j\in \A\setminus \O} d(\sigma(j),\pi(j)) \right] \right).
\]

\begin{definition}[Truncated distance \cite{GM09}]
For $\tau\geq 0$, define $\LL_\tau:\PP\times \PP\to \R$ as $\LL_{\tau}(u,v) = \max\{d(u,v) - \tau, 0 \}$ and $\rho_\tau:\A\times \PP \to \R$ as $\rho_\tau(j, u) = \E_{\sigma}[\LL_{\tau}(\sigma(j),u)]$. Note that $\LL_{\tau}(\cdot,\cdot)$ is not a metric for $\tau >0$.
\end{definition}

\begin{definition} 
Given a node set $Z\subseteq \A$, let $\PP(Z) \subseteq \PP$ be 
the associated point set corresponding to possible realizations of nodes in $Z$. 
Let $\sol(Z,k,t,\rho_\tau)$ and $\opt(Z,k,t,\rho_\tau)$ be a solution by algorithm and the global optimum solution respectively to the
$(k,t)$-median problem on node set $Z$ where the centers are restricted to {$\PP(Z)$} and the weighted assignment 
cost of assigning node $j \in Z$ to center $m\in\PP(Z)$ is $\rho_\tau(j,m)$. The costs $\Csol(Z,k,t,\rho_\tau)$ and $\Copt(Z,k,t,\rho_\tau)$ are defined analogously.
\end{definition}

Let $d_{\min}$ and $d_{\max}$ denote the minimum and the maximum distance, respectively, between two distinct points in $\PP$ and let $\Delta = d_{\max}/d_{\min}$. The algorithm is presented in Algorithm~\ref{alg:center-g}. 
\begin{algorithm}[H]
\begin{algorithmic}[1]
\State All parties compute $d_{\min}$ and $d_{\max}$
\State Each party creates $\mathbb{T} = \{2^i d_{\min}/18: 0\leq i\leq \lceil \log_2\Delta\rceil + 2\}$
\For{each $\tau\in \mathbb{T}$}
\State All parties run Algorithm~\ref{alg:involved2b} with the following changes: when it calls Algorithm~\ref{alg:involved2a} as a subroutine, $\sol(\A_i,2k,q)$ in Algorithm~\ref{alg:involved2a} is replaced with $\sol(\A_i,2k,q,\rho_{6\tau})$ and the sites obtain the numbers of local outliers $\{t_i(\tau)\}$ \label{alg:g:main}
\EndFor
\State Coordinator finds $\hat\tau\!=\!\min\{\tau\!\in\!\mathbb{T}:\!\sum_i \Csol(\A_i, 2k, t_i(\tau), \rho_{6\tau})\!\leq\!12\tau\}$ \label{alg:g:tau_star}
\State Coordinator solves $(k,t)$-center-g on the preclustering solutions $\sol(\A_i, 2k, t_i(\hat\tau),\rho_{6\tau})$ and outputs $\sol(\A,k,(1+\epsilon)t)$.
\end{algorithmic}
\caption{Algorithm for $(k,t)$-center-g}
\label{alg:center-g}
\end{algorithm}

Now we try to analyze the performance of Algorithm~\ref{alg:center-g}. We first show an analogy of Theorem~\ref{thm:alt-median} that we can compute a constant approximation to $\Copt(Z,k,t,\rho_\tau)$. The proof is omitted.  

\begin{lemma}
\label{needthis2}\label{NEEDTHIS2}
Let $\tau \geq 0$. For the $(k,t)$-center problem on $Z$, we can compute in $\tilde O((k+t)|Z|)$ time $\sol(Z,k,(1+\epsilon)t,\rho_{9\tau})$ or $\sol(Z,(1+\epsilon)k,t,\rho_{3\tau})$ such that
\begin{gather*}
\Csol(Z,k,(1+\epsilon)t,\rho_{9\tau})\leq \max\{6,6/\epsilon\}\cdot\Copt(Z,k,t,\rho_\tau)\\
\Csol(Z,(1+\epsilon)k,t,\rho_{3\tau})\leq \max\{6,6/\epsilon\}\cdot\Copt(Z,k,t,\rho_\tau)
\end{gather*}
\end{lemma}

We next show that the $\hat\tau$ computed in Step~\ref{alg:g:tau_star} is a good choice of $\tau$ and will ensure that the preclustering solutions $\sol(\A_i,2k,t_i(\hat\tau),\rho_{2\hat\tau})$ can be combined to yield a good global solution. Specifically we have the following two lemmas.

\begin{lemma}\label{lem:implement}
The $\hat\tau$ computed in Step~\ref{alg:g:tau_star} satisfies the following two conditions.
\begin{enumerate}[(i)]\parskip=0in
\item  $\sum_i \Csol(\A_i,2k,t_i(\hat\tau),\rho_{6\hat\tau}) \leq 12\hat\tau$;
\item  $\sum_i \Copt(\A_i,k,t_i',\rho_{2\hat\tau}) \geq 2\hat\tau$ for all $\{t_i'\}$ s.t. $\sum_i t_i' \leq t$,
\end{enumerate}
\end{lemma}
\begin{proof}
Note that $\tau_{\max} = \max \mathbb{T} > d_{\max}/6$, it always holds that $\rho_{6\tau_{\max}} = 0$. Thus the condition $\sum_i \Csol(\A_i,2k,t_i(\tau_{\max}),\rho_{6\tau_{\max}})\leq 12\tau_{\max}$ holds, and $\hat\tau$ exists and satisfies condition (i).

Next we show that condition (ii) holds. Let $\{t_i'\}$ be an arbitrary sequence satisfying that $\sum_i t_i' \leq t$. Similarly to the proof of Lemma~\ref{lem:waterfilling}, one can show that $\sum_i \Csol(\A_i, 2k, t_i', \rho_{6\hat\tau})\geq \sum_i \Csol(\A_i, 2k, t_i(\hat\tau), \rho_{6\hat\tau})$, using the fact that $\sum_i t_i'\leq t < \rho t = \sum_i t_i$. Combining with Lemma~\ref{needthis2} with $\epsilon = 1$, we have that
\begin{multline*} 6 \sum_i \Copt(\A_i,k,t_i',\rho_{2\hat\tau}) \geq 
\sum_i \Csol(\A_i,2k,t_i',\rho_{6\hat\tau}) \\
\geq \sum_i \Csol(\A_i,2k,t_i(\hat\tau),\rho_{6\hat\tau}) \geq 12\hat\tau,
\end{multline*}
whence condition (ii) follows.
\end{proof}

\begin{lemma}\label{lem:combine_global}
Suppose that $\hat\tau$ satisfies the condition (i) and (ii) of Lemma~\ref{lem:implement}, a $\gamma$-approximation of the weighted center-g problem induced by preclustering $\sol(\A_i,2k,t_i(\hat\tau),\rho_{6\hat\tau})$ is an $O(\gamma)$ approximation of $\Copt(\A,k,t)$.
\end{lemma}
To prove this lemma, we need the following two auxiliary lemmas.
\begin{lemma}
\label{lem1-center-un}
$2 \Copt(\A,k,t,\rho_{\tau}) \geq \sum_i \Copt(\A_i,k,t_i^\ast,\rho_{2\tau})$, where $t_i^\ast$ is the number of ignored nodes from $\A_i$ in the global optimum solution $\opt(\A,k,t,\rho_{\tau})$.
\end{lemma}
\begin{proof}
Fix a realization of the nodes.
The proof mimics Lemma~\ref{lem1} for each realization. It then uses the observation that
$\LL_{\tau}(u_1,u_2)+\LL_{\tau}(u_2,u_3) \geq \LL_{2\tau}(u_1,u_3)$ and takes the expectation. 
\end{proof}

\begin{lemma}  \label{yikeslemma}
If $\Copt(Z, k,t,\rho_{\tau}) \geq \tau$ then $\Copt(Z,k,t) \geq \tau/3$.
\end{lemma}
\begin{proof}
The case of $t = 0$ (no outliers) is proved in \cite[Lemma 4.4]{GM09}. For a general $t > 0$, let $Z'\subseteq Z$ be the set of clustered point in $\opt(Z,k,t)$, then $\Copt(Z',k,0,\rho_{\tau}) = \Copt(Z,k,t,\rho_\tau)\geq \tau$, thus $\Copt(Z,k,t)=\Copt(Z',k,0)\geq \tau/3$.
\end{proof}

\begin{proof}[Proof of Lemma~\ref{lem:combine_global}]
It follows from Lemma~\ref{lem1-center-un} and condition (ii) of Lemma~\ref{lem:implement} that 
\[
2\Copt(\A,k,t,\rho_{\hat\tau}) \geq \sum_i \Copt(\A_i,k,t_i^\ast,\rho_{2\hat\tau}) \geq 2\hat\tau,
\]
where $t_i^\ast$ is the number of ignored nodes from $\A_i$ in the global optimum solution $\opt(\A,k,t,\rho_{\hat\tau})$. It then follows from Lemma~\ref{yikeslemma} that $\Copt(\A,k,t)\geq \hat\tau/3$,

To simplify the notation, in the rest of the proof we shorthand $t_i(\hat\tau)$ as $t_i$. Let $\A^*_i\subseteq \A_i$ be the set of nodes clustered in the global optimum solution $\opt(\A,k,t)$.
Consider ``collapsing'' the nodes in $\A^*_i$ to their
corresponding centers in $\sol(\A_i,2k,t_i,\rho_{6\hat\tau})$ while keeping the same centers in
 $\sol(\A,k,t)$.
If a node in $\A^*_i$ is marked as an outlier in
 $\sol(\A_i,2k,t_i,\rho_{2\hat\tau})$ then it is not moved, 
and it continues to be excluded from the calculation.
This movement increases the expectation of the maximum assignment
by $6\hat\tau + \Csol(\A_i,2k,t_i,\rho_{2\hat\tau})$. 
Now consider the same process where we
collapse $\A^*_i$ for all $i$. The total increase across the different $i$ is $6\hat\tau + \sum_i
\Csol(\A_i,2k,t_i,\rho_{6\hat\tau})$ because the increase in $6\hat\tau$ arises from distance truncation and is common.
Thus we achieve a solution of cost at most 
\[\gamma \left( \Copt(\A,k,t) + 6\hat\tau + \sum_i \Csol(\A_i,2k,t_i,\rho_{6\hat\tau}) \right).\]
Now consider ``expanding'' the nodes of $\A_i$ from the preclustering to the distribution $\D_j$. By that logic the expected maximum can increase by 
at most $2\hat\tau + \sum_i \Csol(\A_i,2k,t_i,\rho_{2\hat\tau})$, which by condition (i) of Lemma~\ref{lem:implement} totals to $O(\gamma\hat\tau) = O(\gamma)\Copt(\A,k,t)$. The lemma follows.
\end{proof}
We state the main theorem for the $(k,t)$-center-g problem to conclude this section.

\begin{table*}[!htbp]
\centering
{\scriptsize
\begin{tabular}{|c|c|c|c|c|c|c|}
\hline
Objective &  Approx. &  Centers & Ignored & Rounds & Total Comm. & Local Time\\
\hline 
\rule{0pt}{3ex}
\multirow{3}{*}{median} & \multirow{3}{*}{$O(1)$} & 
\multirow{3}{*}{$k$} & \multirow{2}{*}{$t$} & 1 &  $\tilde O((sk + st)B)$ & $\tilde{O}(n^2_i)$, $\tilde{O}(k^2s^3t^5)$\\
& & & & 2 &  $\tilde O((sk + t)B)$ & 
 $\tilde{O}(n^2_i)$, $\tilde{O}(k^2t^2(sk+t)^3)$\\
 \cline{4-7}
\rule{0pt}{3ex} & & & $(2+\delta)t$ & 2 & $\tilde{O}(s/\delta +skB)$ & $\tilde{O}(n^2_i)$, $\tilde{O}(s^2k^7)$\\
\hline
\rule{0pt}{3ex}
\multirow{4}{*}{\begin{tabular}{c} means/\\median\end{tabular}} 
& \multirow{4}{*}{$O(1+1/\epsilon)$} & \multicolumn{2}{c|}{ \multirow{2}{*}{$k,(1+\epsilon)t$ or $(1+\epsilon)k,t$}} & 1 & $\tilde O((sk + st)B)$ & $\tilde{O}(n^2_i)$, $\tilde{O}((sk+st)^2)$ \\
& & \multicolumn{2}{c|}{} & 2 & $\tilde O((sk + t)B)$  & $\tilde{O}(n^2_i)$, $\tilde{O}((sk+t)^2)$ \\					
\cline{3-7}
\rule{0pt}{3ex} &  &   $k$ & $(2+\epsilon+\delta)t$& \multirow{2}{*}{2} & \multirow{2}{*}{$\tilde O(s/\delta + skB)$}  & 
\multirow{2}{*}{$\tilde{O}(n^2_i)$, $\tilde{O}((sk)^2)$}\\
& & $(1+\epsilon)k$ & $(2+\delta)t$ &  & & \\
\hline
\rule{0pt}{3ex}
\multirow{3}{*}{center} & \multirow{3}{*}{$O(1)$} &  
\multirow{3}{*}{$k$} & \multirow{2}{*}{$t$} & 
1 & $\tilde O((sk + st)B)$ & $\tilde{O}((k+t)n_i),\ \tilde{O}((sk+st)^2)$\\
& & & & 2 & $\tilde O((sk +t)B)$ & $\tilde{O}((k+t)n_i),\ \tilde{O}((sk+t)^2)$\\
\cline{4-7}	\rule{0pt}{3ex}					    	&  &  & $(2+\delta)t$ & 2 & $\tilde O(s/\delta+skB)$ & $\tilde{O}(n^2_i),\ \tilde{O}((sk)^2)$\\
\hline
\begin{tabular}[x]{@{}c@{}}\rule{0pt}{2.5ex} uncertain\\
median/\\means/\\center-pp \end{tabular} & 
\multicolumn{5}{c|}{{as in the regular case above}} & {$+ O(n_iT)$, unchanged}\\
\hline
\rule{0pt}{3ex}
\eat{
\multirow{5}{*}{center-g} & \multirow{5}{*}{$O(1+\frac{1}{\epsilon})$} &
\multicolumn{2}{c|}{ \multirow{3}{*}{$k,(1+\epsilon)t$ or $(1+\epsilon)k,t$}} & 1 & $\tilde O(s((k + t) \log \Delta + tI))$ & $\tilde{O}(n^2)$, $\tilde{O}((sk+st)^2)$ \\
& & \multicolumn{2}{c|}{} & 2 & $\tilde O(sk + stI)$  & $\tilde{O}(n^2)$, $\tilde{O}((sk+st)^2)$ \\
& & \multicolumn{2}{c|}{} & 2 & $\tilde O(s \log \Delta + tI)$  & $\tilde{O}(n^2)$, $\tilde{O}((sk+t)^2)$ \\						
\cline{3-7}
\rule{0pt}{3ex} &  &   $k$ & $(2+\epsilon+\delta)t$& \multirow{2}{*}{2} & \multirow{2}{*}{$\tilde O((s/\delta) \log \Delta + sk)$}  & 
\multirow{2}{*}{$\tilde{O}(n^2)$, $\tilde{O}((sk)^2)$}\\
& & $(1+\epsilon)k$ & $(2+\delta)t$ &  & & \\
}
\rule{0pt}{3ex}
\multirow{2}{*}{center-g} & $O(1 + 1/\eps)$ & \multirow{2}{*}{$k$} & $(1+\eps)t$ & 2 & $\tilde O(skB + tI + s\log \Delta)$  & $\tilde{O}(n^2_i \log \Delta)$, $\tilde{O}((sk+t)^2)$  \\
\rule{0pt}{3ex}
& $O(1)$ & & $t$ & 1 & $\tilde O(s(kB+tI)\log\Delta)$ &  $\tilde{O}((k+t)n_i \log\Delta)$,  $\tilde{O}(s^2(k+t)^2)$ \\
\hline
\end{tabular}
}
\caption{Our results. $T$ denotes the runtime to compute $1$-median/mean of a node distribution, $I$ is the information encoding a node in the uncertain data case, $B$ the information encoding a point and  $\Delta$ the ratio between the maximum pairwise distance and the minimum pairwise distance in the dataset.}
\label{tab:our_results}
\end{table*}

\begin{theorem}
\label{thm:center-g}
For the distributed $(k,t)$-center-g problem, Algorithm~\ref{alg:center-g} outputs $\sol(\A,k,(1+\epsilon)t)$ satisfying $\Csol(\A,k,(1+\epsilon)t) = O(1+1/\epsilon)\cdot\Copt(\A,k,t)$. The sites communicate a total of $\tilde O(skB+ s\log\Delta+tI)$ bits of information to the coordinator over $2$ rounds, where $I$ is the bit complexity to encode a node. The runtime at site $i$ is $\tilde O((k+t)n_i\log\Delta)$ and the runtime at the coordinator is $\tilde O((sk+t)^2)$.
\end{theorem}
\begin{proof}
The claim on approximation ratio follows from Lemma~\ref{lem:combine_global}. To determine $\hat\tau$, the communication cost increases by a factor of $\log\Delta$; to send the preclustering solutions, the communication cost for sending the outliers increases by a factor of $I$. The runtime follows from Lemma~\ref{needthis2} with an increase of a factor of $\log \Delta$.
\end{proof}
We remark that the dependence on $\log \Delta$ can be removed with another pass where each site computes a $\tau_i$ using binary search. The discussion is omitted in the interest of simplicity. 

Other results claimed in Table~\ref{tab:our_results} follow from analogous amendments to Theorem~\ref{thm:clustering_only}.

\bibliographystyle{ACM-Reference-Format}
\bibliography{literature}

\balance
\appendix
\section{The Full Set of Our Results}
\label{sec:full-results}

We summarize the full set of our results in Table~\ref{tab:our_results}.  Besides the main results that already appear in Table~\ref{tab:main-results}, all the $1$-round results in Table~\ref{tab:our_results} basically follow from setting $t_i=t$ for all sites $i$.   The results for $(k,t)$-median/means that ignore $(2+\delta)t$ or $(2+\eps+\delta)t$ points basically follow from Theorem~\ref{thm:clustering_only}, where for $(k,t)$-median with $k$ centers (unicriterion) we need to apply again the $1$-round result, and for $(k,t)$-median/means with $(1+\eps)k$ centers we simply use the second inequality of Theorem~\ref{thm:alt-median} instead of the first one at the final clustering step at the coordinator.  The result for $(k,t)$-center that ignore $(2+\delta)t$ points is due to the following modifications on Algorithm~\ref{sec:center1}: sites do not send the total $(1+\delta)t$ local outliers to the coordinator, and thereafter the coordinator performs the second level clustering with (another) $t$ outliers, we have $(2+\delta)t$ outliers in total.

\section{Proof of Theorem~\ref{thm:alt-median}}
\label{proof:altmedian}

\begin{proof}
The result in \cite{CKMN01} prioritized approximation ratio and used
\cite{CG99} instead of \cite{JV01}. However the former increases the
running time to $\tilde{O}(n^3)$ to get the better approximation
factor. Using the latter result of \cite{JV01}, we get the running time of the
first part of the theorem. To observe the quality guarantee, note that
\cite{JV01} creates two solutions with $k_1$, $k_2$ centers and each solution ignores exactly $t$ outliers, where $k_1 <
k < k_2$. Although not explicitly stated in \cite{JV01}, but as observed in \cite{CKMN01}, the algorithm is applicable
to the outlier case as we can simply stop the algorithm when there are $t$ points unprocessed.

Set $a=(k_2 - k)/(k_2-k_1)$ and $b=1-a$ and consider
the convex combination of the two solutions. The convex
combination of their costs is a $3$-approximation. To get exactly $k$
centers, we iteratively pair off every center in the small solution
with its nearest (remaining) center in the large solution. With
probability $a$ we choose all the centers in the small solution and
otherwise we choose the paired centers in the large solution. In the latter case we have chosen $k_1$
centers and we choose the remaining $k-k_1$ centers at random from the
remaining centers in the large solution. Note that every center in the large solution is chosen with probability at least $1-a$.

In the current case we also have two solutions and each 
solution ignores exactly $t$ outliers. Notice that if a point is
labeled outlier in one solution and not in the other, it must be
directly connected to a center (in the language of \cite{JV01}). In
the case we choose the centers in the small solution, all the
points that were directly connected to the center continue to
satisfy that $6$ times its dual value is greater than the distance to its
center plus $6$ times its payment towards the centers
(\cite{CKMN01} reduced this to $4$ based on \cite{CG99}).
If we choose all the centers in the small solution then 
we cannot have more than $t$ outliers. If we choose the large solution 
then we may exceed $t$ outliers if all the points labeled outliers in the 
small solution were excluded and some of the points clustered in the large solution 
(but not in the small) cannot be accommodated because the corresponding center was 
not chosen. But this happens with probability at most $1-(1-a)=a$ and therefore in 
expectation we lose an extra $a\cdot t$ outliers. If $a \geq \epsilon/2$, we choose the small 
solution, which provides $6/\epsilon$ approximation. Otherwise we run the rounding  
part in \cite{JV01} multiple times and choose a solution with at most $t + \epsilon t$ outliers 
(which happens with probability $O(\epsilon)$ using Markov inequality).

For the second part observe that if $k_1+k_2 \geq (1+\epsilon)k$ then $a>\epsilon$ and we
have $3/\epsilon$ approximation with $k$ centers. Otherwise we use {\em all}
the $k_1+k_2$ centers. Now we can assert that the distance cost plus
$3$ times the cost towards centers is at most $3$ times the dual
value for all points not marked as outliers in both
solutions. Thus the total number of outliers will be the intersection
of the outliers of the two solutions and at most $t$.  The theorem
follows.

Note that the above rounding argument uses triangle inequality. While the triangle inequality does not hold for squares of distances (as in the $k$-means objective function), we instead use $2(a^2 + b^2) \geq (a+b)^2$.
\end{proof}

\section{Proof of Lemma~\ref{needthis2}}\label{sec:bicriteria_tau}
\begin{proof}
The proof is similar to that of Theorem~\ref{thm:alt-median}. The only different part is the 
accounting for the truncation. For the $(1+\epsilon)k$ result we note a pseudo-triangle inequality  (see \cite[Lemma 4.1]{GM09})
$\rho_{3\tau}(j,m)\leq \rho_{\tau}(j,m')+\rho_{\tau}(i,m') + \rho_{\tau}(i,m)$ for any $m'$, 
since in this case we assign points within three hops. For the $(1+\epsilon)t$ result we assign within $9$ hops---each point has a center in the large and small solutions within $3$ hops. The pairing of the centers in the two solutions show that the pair of a center in the small solution exists within $6$ hops. The whole argument for Theorem~\ref{thm:alt-median} then goes through.
\end{proof}

\end{document}